\documentclass[a4paper,12pt]{scrartcl}

\usepackage{amsmath}
\usepackage{amssymb}
\usepackage{amsthm}
\usepackage{stmaryrd}
\usepackage{bbm}
\usepackage[colorlinks]{hyperref}
\usepackage{mathtools}
\usepackage[numbers]{natbib}
\usepackage{enumitem}
\usepackage{subcaption}
\usepackage{float}
\usepackage{authblk}
\usepackage{tikz}
\usepackage{tikz-3dplot}
\usetikzlibrary{3d,calc}
\tdplotsetmaincoords{70}{120} 
\usetikzlibrary{arrows,automata}
\usepackage{setspace}
\usepackage{graphicx}
\usepackage{subcaption} 

\usepackage{todonotes} 

\newenvironment{remark}[1][Remark]{\begin{trivlist}
\item[\hskip \labelsep {\bfseries #1}]}{\end{trivlist}}

\theoremstyle{plain}
\newtheorem{theorem}{Theorem}[section]  
\newtheorem{proposition}[theorem]{Proposition}

\theoremstyle{definition}

\theoremstyle{remark}

\numberwithin{equation}{section}

\title{Discrete-Time Two-Strain Epidemic Dynamics on Complex Networks}

\author[1]{\normalsize Frank Namugera}
\affil[1]{\itshape Mathematics Department, Makerere University, Kampala, Uganda}
\affil[ ]{Email: \texttt{frank.namugera@mak.ac.ug} \vspace{-3em}}

\date{}
\begin{document}
\maketitle

\begin{abstract}

We investigate a discrete-time two-strain symbiotic epidemic model on complex networks with both random and long-range interactions. Our analysis examines how the co-infection recovery rate ($\mu$), the long-range decay exponent ($\alpha$), and the scale-free connectivity exponent ($\gamma$) shape epidemic persistence under cooperative dynamics. Comparison with a two-strain competition model shows how these parameters control strain dominance, coexistence, or extinction.

The results demonstrate that contagion dynamics are strongly affected by environmental randomness and long-range couplings. In facultative symbiosis, the co-infection recovery rate undergoes a clear phase transition, separating persistence from extinction. In the competitive setting, regimes with $\alpha < 2$ and $\gamma < 3$ markedly lower the epidemic threshold, allowing persistence even at small contagion rates ($\sigma$). Statistical analysis further reveals that $\gamma$ and $\alpha$ exert pronounced, nonlinear, and time-dependent effects on strain survival.

\end{abstract}

\textbf{Key words}: Co-infection, Symbiosis, Epidemic Spreading, Random environment, Long-range decay, Complex networks.


\section{Introduction}

Network-based epidemic models have become fundamental in studying contagions for capturing how spatial randomness~\cite{gomez2010discrete}, long-range interactions~\cite{estrada2011epidemic,namugera2025long}, and stochastic transmission shape transition dynamics. Beyond their theoretical significance, these models have proven essential in guiding public health interventions, as underscored during the COVID-19 pandemic, where controlling contact was one of the key strategies~\cite{he2020covid}. A central insight is that epidemic outcomes are deeply constrained by the topology and heterogeneity of the underlying contact network~\cite{pastor2001epidemic,newman2002spread,newman2010networks}, which governs both threshold behavior and large-scale spreading patterns.

While classical epidemic models have focused mainly on single-strain dynamics, real-world epidemics often involve confounding or interacting circulating pathogens. For instance, co-infections, which are simultaneous infections by multiple strains, are known to modulate transmissibility, recovery, and disease severity~\cite{oka2020spatial,yang2021epidemic,Namugera2025contact}. Moreover,  interactions may be competitive, where one pathogen suppresses another, or synergistic, where co-infection increases transmission or virulence. Despite their epidemiological relevance, most existing models treat co-infection in simplified settings, typically on homogeneous or locally connected graphs~\cite{funk2010interacting, susi2015co, mclean2018consequences, durrett2020symbiotic}.

In this work, we propose two fundamental modes of multi-strain interaction on a complex network: 1) the symbiosis model, where all strains benefit mutually from co-infection, and 2) the competition model, where the presence of one strain suppresses the survival of the other. Both interactions have been extensively investigated on simple graphs that do not incorporate range dynamics. For instance, we showed in~\cite{Namugera2025contact} that the epidemic threshold in a symbiotic model depends strongly on the independent thresholds of each strain. Furthermore, we showed the existence of a co-infection threshold, above which strains that would otherwise be subcritical can persist in nearest-neighbor interaction settings. Although prior work has investigated cooperative symbiosis, where the presence of one strain enhances the spread of another~\cite{chen2013outbreaks}, our focus is on the independent-spread case, where multiple strains may occupy the same node without directly facilitating one another. This accounts for facultative positive cooperation, where the delayed recovery may reduce the epidemic thresholds of the individual strains.

Existing studies indicate that co-infection processes exhibit distinctive transition behaviors, yet the role of long-range interactions remains less  explored. Environmental heterogeneity is also known to strongly influence epidemic dynamics~\cite{sanz2014dynamics, hebert2015complexity}. Building on this, Cui et al.~\cite{cui2017mutually,cui2019epidemic} extended co-infection models to heterogeneous and multiplex networks, emphasizing cooperative effects; however, their analyses were not developed within a discrete-time framework.

On the other hand, in a perfect competition model, which assumes a winner-takes-all mechanism, the system exhibits three distinct regimes: coexistence of both strains, domination by a single strain, or extinction of all strains. This problem has been rigorously analyzed in lattice-based nearest-neighbor models, most notably by Doshi et al.~\cite{doshi2021competing}, who provide a full characterization of the competitive dynamics. Beutel et al.~\cite{beutel2012interacting} identify threshold conditions that govern transitions between coexistence and competitive exclusion, while Newman~\cite{newman2005threshold} demonstrates that coexistence is possible only within a narrow intermediate range of transmissibility for the dominant pathogen. These results, largely derived from ODE-based and nearest-neighbor frameworks, highlight the presence of threshold phenomena and competition-induced phase transitions. This serves as a reference point for extending the analysis to long-range and stochastic epidemic models.

We develop a discrete-time epidemic model with long-range interactions in a random environment. Unlike continuous-time formulations, we use difference equations to capture the system's temporal evolution more explicitly. While continuous-time models often average out short-term fluctuations, the discrete-time framework preserves finer details of the dynamics, enhancing interpretability. Moreover, by incorporating long-range interactions and environmental randomness, the model better reflects population-level heterogeneity and demographics, making it more applicable than nearest-neighbor formulations.

This work is, in one way, an extension of the nearest neighbour case in ~\cite{Namugera2025contact}, allowing for long-range transmission as studied by ~\cite{estrada2011epidemic}, and enabling environmental heterogeneity, also covered by \cite{gomez2010discrete}. 
We begin by describing the environment where the transmission occurs through long-range interactions governed by a distance-weighted kernel:
\[
p_{ij} \propto \frac{r_{ij}}{d_{ij}^\alpha},
\]
where \( d_{ij} \) denotes the shortest-path distance between nodes \( i \) and \( j \), and \( \alpha > 0 \) is a parameter controlling the strength of spatial decay. The term \( r_{ij} \) captures the probability that node \( i \) successfully transmits to node \( j \) through multiple stochastic attempts, and is given by
\[
r_{ij} = 1 - \left( 1 - \frac{w_{ij}}{w_i} \right)^{\lambda_i},
\]
where \( \lambda_i \) is the number of infection attempts originating from node \( i \), \( w_{ij} \) is the edge weight from \( i \) to \( j \) drawn from a random environment, and \( w_i = \sum_k w_{ik} \) is the total outgoing weight from node \( i \). The environmental variability encoded in \( w_{ij} \) is assumed to follow a uniform distribution, reflecting random heterogeneity in interaction strength. This formulation generalizes distance-based epidemic models~\cite{juhasz2015longrange} and aligns with recent approaches to mobility-driven contagion~\cite{he2020covid}. For this model, we consider the case where $\lambda_i = 1$. This means that for every time $t$, there is one possible contact. Then, the environment reduces to $r_{ij} = \omega_{ij}$ if we do not consider normalization at each node.

To analyze these phenomena, we study the spectral properties of the effective transmission operator and derive thresholds for epidemic survival. Numerical simulations validate these theoretical predictions and reveal nontrivial survival behavior driven by the interplay of coinfection dynamics, network structure, and spatial decay. Overall, this work presents a comprehensive and adaptable framework for understanding multi-strain epidemic processes in spatially structured and stochastic environments.

\paragraph{Structure of the paper:} Section 2 explores the mathematical structure of the two models. First, in section 2.1, we examine the co-infection model, its analysis, and epidemic threshold. Later, in section 2.5, we compare this to the competition model and discuss its stability conditions. In section 3, we perform numerical simulations for both models. We analyze the results and conclude with final remarks in section 4. 

\section{Model Description}
\subsection{Co-infection model}

We describe a two-strain cooperative SIS (Susceptible, Infected, Susceptible) contact process, where recovery is temporary and no permanent immunity is possible. The population is represented by a finite graph \( G = (\mathcal{V}, \mathcal{E}) \), where each vertex \( i \in \mathcal{V} \) denotes an individual and edges \( \mathcal{E} \) specify possible transmission pathways. Each node occupies one of four epidemic states: susceptible (\(S\)), infected with strain \(A\), infected with strain \(B\), or coinfected (\(AB\)), forming the state space \( \mathcal{S} = \{S, A, B, AB\} \). The probabilities that node \(i\) is in states \(A\), \(B\), or \(AB\) at discrete time \(t\) are denoted \(p_{i,t}^A\), \(p_{i,t}^B\), and \(p_{i,t}^{AB}\), respectively. The system evolves according to a nonlinear stochastic difference equation incorporating both transmission and recovery. Nodes infected with either strain recover at rate \(\delta\), while co-infected nodes recover more slowly at rate \(\mu < \delta\). 

The update rules governing the infection probabilities are
\begin{align}
p_{i,t+1}^A &= (1 - p_{i,t}^A)(1 - \zeta_{i,t+1}^A) + p_{i,t}^A (1-\delta)\zeta_{i,t+1}^B + \tfrac{1}{2}\mu p_{i,t}^{AB}, \label{eq:Aupdate} \\
p_{i,t+1}^B &= (1 - p_{i,t}^B)(1 - \zeta_{i,t+1}^B) + p_{i,t}^B (1-\delta)\zeta_{i,t+1}^A + \tfrac{1}{2}\mu p_{i,t}^{AB}, \label{eq:Bupdate} \\
p_{i,t+1}^{AB} &= (1 - p_{i,t}^{AB})(1 - \zeta_{i,t+1}^{AB}) + (1-\mu)p_{i,t}^{AB}, \label{eq:ABupdate}
\end{align}
where the last term in \eqref{eq:Aupdate}–\eqref{eq:Bupdate} reflects transitions from co-infection back to single-strain infection, and the structure of \eqref{eq:ABupdate} captures both persistence and rapid reinfection of coinfected nodes.

The probability that node \(i\) avoids infection at time \(t+1\) is denoted by $\zeta_{i,t+1}^K$ where $K$ is the strain being avoided. Hence,
\begin{align}
\zeta_{i,t+1}^A &= \prod_{j \in \mathcal{N}_{LR}(i)} \big(1 - \beta_{ij}^A p_{j,t}^A\big)\big(1 - \beta_{ij}^{AB} p_{j,t}^{AB}\big), \\
\zeta_{i,t+1}^B &= \prod_{j \in \mathcal{N}_{LR}(i)} \big(1 - \beta_{ij}^B p_{j,t}^B\big)\big(1 - \beta_{ij}^{AB} p_{j,t}^{AB}\big), \\
\zeta_{i,t+1}^{AB} &= \big[1 - p_{i,t}^A(1-\zeta_{i,t+1}^B)\big]\big[1 - p_{i,t}^B(1-\zeta_{i,t+1}^A)\big], \label{eq:zetaab}
\end{align}
where
\begin{equation}
 \label{eq:beta}
\beta_{ij}^K = \sigma_K\, r_{ij}\, d_{ij}^{-\alpha}, \quad K \in \{A,B,AB\}. 
\end{equation}

Here, \(\sigma_K\) is the strain-specific infection strength.
This formulation places the model within the broader class of stochastic interacting particle systems with long-range couplings, thereby linking epidemic propagation to nonequilibrium phase transitions. In contrast to our earlier work~\cite{Namugera2025contact}, the present framework excludes instantaneous reinfection and introduces a redefined expression for the probability of avoiding a co-infection.

\subsection{Nonlinear Analysis}

We now analyze the discrete-time dynamics of the nonlinear contact process defined in Eqs.~\eqref{eq:Aupdate}–\eqref{eq:zetaab}. Let the infection state at time \( t \) be:
\[
\mathbf{p}_t := \left( p_{i,t}^A, p_{i,t}^B, p_{i,t}^{AB} \right)_{i \in \mathcal{V}} \in [0,1]^{3n},
\]
where \( n = |\mathcal{V}| \) is the number of nodes. Define the feasible state space as the compact hypercube \( \Omega := [0,1]^{3n} \). The first proposition examines whether the model defines a valid probability measure. Given the presence of multiple parameters, it is essential to verify that the model is well-posed. To this end, we assess the continuity of the update function and the forward-invariance of the state space to ensure that probabilities remain within a valid range over time.

\begin{proposition}[Well-Posedness of the Dynamics]
Let \( \mathbf{F} : \Omega \to \Omega \) be the update map defined by the system \eqref{eq:Aupdate}–\eqref{eq:ABupdate}. Then:
\begin{enumerate}
    \item \emph{(Continuity)} The map \( \mathbf{F} \) is continuous on \( \Omega \).
    \item \emph{(Forward-Invariance)} The image of \( \Omega \) is contained in \( \Omega \), i.e., \( \mathbf{F}(\Omega) \subseteq \Omega \).
\end{enumerate}
\end{proposition}

\begin{proof} {}

\textbf{(1) Continuity:} 

Each component of \( \mathbf{F} \) is constructed from the update rules \eqref{eq:Aupdate}–\eqref{eq:ABupdate}, which involve addition, multiplication, and scalar multiplication of the infection probabilities and transmission terms. The transmission probabilities \(\zeta_{i,t+1}^A\) and \(\zeta_{i,t+1}^B\) are finite products of terms of the form \((1 - \beta_{ij}^K p_{j,t}^K)(1 - \beta_{ij}^{AB} p_{j,t}^{AB})\), while \(\zeta_{i,t+1}^{AB}\) is given by a product of linear functions of \(p_{i,t}^A\) and \(p_{i,t}^B\). Since all coefficients lie in \([0,1]\), each factor is continuous, and finite sums and products of continuous functions are continuous. Substituting these into the update rules shows that every component of \( \mathbf{F} \) is continuous on \(\Omega\).

\textbf{(2) Forward-Invariance.}  
Fix \( i \in \mathcal{V} \). All the current state satisfies \( p_{i,t}^A, p_{i,t}^B, p_{i,t}^{AB} \in [0,1] \). We must show that the updates \eqref{eq:Aupdate}–\eqref{eq:ABupdate} also remain in \([0,1]\).  

Since \( \delta, \mu \in [0,1] \) and each infection kernel satisfies \( \beta_{ij}^K = \sigma_K r_{ij} d_{ij}^{-\alpha} \in [0,1] \), the factors 
\((1 - \beta_{ij}^K p_{j,t}^K)\) and \((1 - \beta_{ij}^{AB} p_{j,t}^{AB})\) lie in \([0,1]\). Consequently, the products defining 
\(\zeta_{i,t+1}^A\), \(\zeta_{i,t+1}^B\), and \(\zeta_{i,t+1}^{AB}\) also belong to \([0,1]\).  

Consider first the update for strain \(A\). The expression  
\[
p_{i,t+1}^A = (1-p_{i,t}^A)(1-\zeta_{i,t+1}^A) 
+ p_{i,t}^A (1-\delta)\zeta_{i,t+1}^B 
+ \tfrac{1}{2}\mu p_{i,t}^{AB}
\]  
is a sum of nonnegative terms, so \(p_{i,t+1}^A \geq 0\). For the upper bound, observe that each term is bounded by its coefficient, giving  
\[
p_{i,t+1}^A \leq (1-p_{i,t}^A) + p_{i,t}^A(1-\delta) + \tfrac{1}{2}\mu.
\]  
Since \((1-p_{i,t}^A) + p_{i,t}^A(1-\delta) \leq 1\) and \(\tfrac{1}{2}\mu \leq \tfrac{1}{2}\), the right-hand side never exceeds 1. Thus \(p_{i,t+1}^A \in [0,1]\). By symmetry, the same argument applies to the update for \(p_{i,t+1}^B\).  

Finally, consider the coinfection update  
\[
p_{i,t+1}^{AB} = (1 - p_{i,t}^{AB})(1 - \zeta_{i,t+1}^{AB}) + (1-\mu)p_{i,t}^{AB}.
\]  
Both summands are products of values in \([0,1]\), and hence nonnegative. Moreover, the first term is bounded by \(1-p_{i,t}^{AB}\) and the second by \(p_{i,t}^{AB}\). Adding these bounds yields at most \(1\). Therefore \(p_{i,t+1}^{AB} \in [0,1]\).  We conclude that every updated component remains in \([0,1]\), and hence the state space \(\Omega\) is forward-invariant under the dynamics: \(\mathbf{F}(\Omega) \subseteq \Omega\).  
\end{proof}

\subsection{Linearization at the Disease-Free Equilibrium (DFE)}

Define the disease-free equilibrium (DFE) as \( \mathbf{p}^* = \mathbf{0} \in \mathbb{R}^{3n} \), corresponding to zero prevalence of all strains. To study local stability of this point, we compute the Jacobian matrix \( J := \nabla g(\mathbf{0}) \) of the discrete-time map \( g(.) \), linearized about \( \mathbf{p}^* = \mathbf{0} \) which contains all first-order partial derivatives of $g(.)$ evaluated at the DFE:
\[
\big[\nabla g(\mathbf{0})\big]_{ij} = 
\frac{\partial g_i}{\partial p_j} \Big|_{\mathbf{p} = \mathbf{0}}, 
\quad i,j = 1, \dots, 3n.
\]

Let \( W_K := (\mathbf{R} \circ \mathbf{D}^{-\alpha} \circ \mathbf{K}) \in \mathbb{R}^{n \times n} \) denote the weighted infection matrix for strain \( K \in \{A, B, AB\} \), where \( \mathbf{R} \in [0,1]^{n \times n} \) contains the random environmental weights, and \( \mathbf{D}^{-\alpha} \in \mathbb{R}^{n \times n} \) encodes the distance-based decay. The Hadamard product is denoted by \( \circ \).

The Jacobian matrix \( J \in \mathbb{R}^{3n \times 3n} \), acting on the state vector \( \mathbf{p} = (p^A, p^B, p^{AB}) \), has the following block upper-triangular form:
\[
J =
\begin{bmatrix}
\sigma_A W_A + (1-\delta) \mathrm{I}_n & 0 & \sigma_{AB} W_{AB} + \frac{1}{2}\mu \mathrm{I}_n\\
0 & \sigma_B W_B + (1-\delta) \mathrm{I}_n &  \sigma_{AB} W_{AB} + \frac{1}{2} \mu \mathrm{I}_n\\
0 & 0 & (1 - \mu) \mathrm{I}_n \\
\end{bmatrix},
\]
where \( W_A \), \( W_B \)  and \( W_{AB} \) describe the linearized infection propagation from neighbors for strains A, B and AB, \( \mu \in [0,1] \) is the rate at which individuals with co-infection resolve into single-strain infections, \( \mathrm{I}_n \) denotes the \( n \times n \) identity matrix.

\subsection{Epidemic Threshold and Stability}
\begin{theorem}[Epidemic Threshold and Local Stability]
\label{thm:threshold}
Let $\lambda_1(W_K)$ be the spectral radius of $W_{K}$, and consider the linearization of the system at the disease-free equilibrium $\mathbf{p}^* = \mathbf{0} \in \mathbb{R}^{3n}$.  
Then $\mathbf{p}^*$ is locally asymptotically stable if 
\[
\frac{\sigma_K}{\delta} < \frac{1}{\lambda_{1}(W_K)}
\]
and is otherwise unstable for $K \in \{A,B\}$.
\end{theorem}

\begin{proof}

We study the stability of the disease-free configuration 
\(\mathbf{p}^* = \mathbf{0}\) 
under the discrete-time evolution
\[
\mathbf{p}_{t+1} = g(\mathbf{p}_t).
\]
The linearized dynamics near $\mathbf{p}^*$ are governed by the Jacobian $(J)$. 

By construction, $J$ is block upper-triangular, with diagonal blocks
\[
J_A = (1-\delta) \mathrm{I}_n + \sigma_{A} W_A, 
\qquad
J_B = (1-\delta) \mathrm{I}_n + \sigma_{B} W_B,
\qquad
J_{AB} = (1-\mu)\mathrm{I}_n,
\]
where $J_A$ and $J_B$ act on the infection probabilities of strains $A$ and $B$, and $J_{AB}$ corresponds to the co-infection component. Since the spectrum of a block upper-triangular matrix is the union of the spectra of its diagonal blocks, then
\[
\operatorname{spec}(J) = \operatorname{spec}(J_A) \cup \operatorname{spec}(J_B) \cup \operatorname{spec}(J_{AB}) .
\]

For discrete-time dynamics, linear stability requires that all eigenvalues of $J$ lie strictly inside the unit disk, i.e.
\[
\rho(J) < 1,
\]
The scalar eigenvalue $1-\mu$ satisfies $|1-\mu| < 1$ whenever $\mu>0$, so it does not influence the stability threshold. Hence the relevant contributions come from the infection blocks $J_A$ and $J_B$.

Since $W_K$ is nonnegative and irreducible, the Perron--Frobenius theorem (\ref{thm:pefro}) implies that its spectral radius is a real, simple eigenvalue $\lambda_{1}(W_K)$. Thus the leading eigenvalue of each block is shifted by the survival term $(1-\delta)$:
\[
\rho(J_K) = (1-\delta) + \lambda_{1}(\sigma_{K} W_K), \qquad K \in \{A,B\}.
\]

Therefore, the disease-free equilibrium is locally asymptotically stable if 
\[
\sigma_{A} \lambda_{1}(W_A) < \delta
\quad \text{and} \quad
\sigma_{B} \lambda_{1}(W_B) < \delta.
\]

If either condition fails, the spectral radius of $J$ exceeds one, leading to exponential growth of small perturbations, and the disease-free state becomes linearly unstable.
\end{proof}

The threshold takes the form \(\tau = 1 / \lambda_1(W_K)\), and persistence occurs if \(\sigma_K / \delta > \tau\). Increasing the spatial decay in the distance kernel $\mathbf{D}^{-\alpha}$ reduces the spectral radius $\lambda_1(W_K)$, thereby raising the epidemic threshold and suppressing disease spread. In contrast, stronger environmental randomness $\mathbf{R}$ amplifies the spectral weight, increasing $\lambda_1(W_K)$ and lowering the threshold. The co-infection parameter $(\mu)$ influences the subsequent nonlinear dynamics but does not affect the initial stability of the disease-free equilibrium. Consequently, the onset of an epidemic is governed by the interplay between transmission-to-recovery ratios, spatial structure, and environmental variability.

\subsection{Contrasting with the competition model}

We now turn to the {competitive regime}, in which colonization by one strain excludes subsequent colonization by its competitor. This ``winner-takes-all'' assumption corresponds to the biologically relevant setting where the first successful infection of a susceptible host prevents superinfection. Retaining the discrete-time framework, the dynamics are given by:
\begin{align}
p_{i,t+1}^A &= (1 - p_{i,t}^A - p_{i,t}^B)(1 - \zeta_{i,t+1}^A) 
+ (1 - \delta_A)\,p_{i,t}^A, \label{eq:A_competition} \\
p_{i,t+1}^B &= (1 - p_{i,t}^A - p_{i,t}^B)(1 - \zeta_{i,t+1}^B) 
+ (1 - \delta_B)\,p_{i,t}^B. \label{eq:B_competition}
\end{align}

Here, \( p_{i,t}^K \) is the probability that node \( i \) is infected by strain \( K \in \{A,B\} \) at time \( t \). The infection pressure is
\begin{align}
\zeta_{i,t+1}^K = \prod_{j \in \mathcal{N}_{LR}(i)} \big(1 - \beta_{ij}^K \, p_{j,t}^K \big), \label{eq:zeta_competition}
\end{align}
with transmission weights \( \beta_{ij}^K \) modulated by environment and distance effects as defined in \autoref{eq:beta}. Unlike in the symbiotic model, recovery is strain-specific, with rates \( \delta_A \) and \( \delta_B \), reflecting differences in clearance or virulence. Furthermore, the infection rate is indexed by strain to show the inhomogeneity in transmission.

\paragraph{Equilibria:}
The system admits up to three classes of fixed points: (i) the disease-free equilibrium (DFE) with \( (p_i^A, p_i^B) = (0,0) \); (ii) single-strain endemic equilibria, where only one strain persists; and (iii) coexistence equilibria, where both strains survive. 

\paragraph{Stability of the DFE:}
Let \( \mathbf{p}_t \in \mathbb{R}^{2n} \) denote the stacked vector of infection probabilities across all nodes, and write
\[
\mathbf{p}_{t+1} = g(\mathbf{p}_t).
\]
Linearizing around \( \mathbf{p^*} = \mathbf{0} \) yields the Jacobian:
\[
\nabla g(0) =
\begin{bmatrix}
(1-\delta_A)\mathbf{I} + \sigma_{A} \mathbf{A}^* & 0 \\
0 & (1-\delta_B)\mathbf{I} + \sigma_{B} \mathbf{B}^*
\end{bmatrix},
\]
where \( \mathbf{A}^* = \mathbf{R}\circ \mathbf{D}^{-\alpha} \circ \mathbf{A} \) encodes weighted long-range connectivity for strain \( A \), and analogously for strain \( B \).  

Define the linear operators
\begin{align}
\mathcal{S}_A &= (1-\delta_A)\mathbf{I} + \sigma_{A} \mathbf{A}^*, \label{eq:comp1} \\
\mathcal{S}_B &= (1-\delta_B)\mathbf{I} + \sigma_{B} \mathbf{B}^*.\label{eq:comp2}
\end{align}
By standard linear stability theory, the DFE is locally asymptotically stable if
\[
\rho(\mathcal{S}_A) < 1 \quad \text{and} \quad \rho(\mathcal{S}_B) < 1,
\]
where \( \rho(\cdot) \) denotes the spectral radius. Following similar steps like in the symbiosis case, Equations \ref{eq:comp1}-\ref{eq:comp2} simplify to 
\[
\frac{\sigma_{A}}{\delta_A} < \frac{1}{\lambda_{1,A^*}}, 
\qquad 
\frac{\sigma_{B}}{\delta_B} < \frac{1}{\lambda_{1,B^*}}.
\]
 where  \( \lambda_{1,A^*}, \lambda_{1,B^*} \) are the respective spectral radii for the strains.
Let $S_K$ denotes the survival score of strain $K \in \{A, B\}$, defined as
\[
S_K := \frac{\sigma_{K}}{\delta_K} \, \lambda_{1, K^*}
\]
then extinction occurs of $S_K<1$ otherwise the strain $K$ survives.

If $S_K < 1$ for all strains $K$, then both infections go extinct. If the condition holds for only one strain, the system reduces to pure competition in which the other strain survives alone. When neither strain satisfies the extinction condition, the outcome depends on nonlinear effects and the underlying network structure, which together determine whether stable coexistence emerges or one strain eventually dominates. These outcomes align with the classification obtained in continuous-time competition models analyzed via monotone dynamical systems \cite{doshi2021competing}.

\section{Numerical results}
\subsection{Symbiosis interation}
Let \( G_A(n, p_1) \), \( G_B(n, p_2) \) and $G_{AB}(n, p_3)$ denote three $one-$dimensional scale-free random graphs representing the contact networks for strains \( A \), \( B \) and $AB$, respectively. Formally, a scale-free graph is a network whose degree distribution follows a power-law form, typically given by
\begin{equation}
    P(k) \sim k^{-\gamma}, 
\end{equation}
where \( P(k) \) is the probability that a randomly chosen node has out degree \( k \), and \( \gamma > 1 \) is the exponent. This suggests that a limited number of nodes (hubs) exhibit exceptionally high degrees of connectivity, whereas the majority display low levels of connectivity, culminating in a heterogeneous network architecture. 

All graphs share the same set of \( n \) nodes (individuals) but differ in their edge orientations, reflecting distinct transmission pathways or contact intensities for each strain. We adopt scale-free random graphs due to their ability to capture heterogeneous contact structures commonly observed in real-world populations, such as heavy-tailed degree distributions, which are known to influence epidemic thresholds and dynamics in nontrivial ways. We summarise the graph characteristics used in the simulations in \autoref{tab:sfgraph}. The spectral radius decreases with increasing values of $\gamma$ and $\alpha$, which raises the epidemic threshold.

\begin{table}[ht!]
\centering
\begin{tabular}{c|ccc|c}
\hline
$\gamma$ & \multicolumn{3}{c|}{$\lambda_{1}$ for different $\alpha$} & Avr. Graph  \\
\cline{2-4}
 & $\alpha = 0.5$ & $\alpha = 1.5$ & $\alpha = 3.0$ & distance \\
\hline
2.0 & 22.37 & 13.59 &  9.30 & 2.57\\
3.5 & 16.20 &  9.16 &  5.63 & 3.08\\
5.0 & 15.84 &  8.51 &  5.03 & 3.13\\
\hline
\end{tabular}
\caption{A scale-free graph with 400 edges and 150 nodes. The average graph distance represents the mean length of the shortest paths between all pairs of nodes, while $\lambda_1$ is the spectral radius simulated over several Monte Carlo simulations.}
\label{tab:sfgraph}
\end{table}

We adopt the simulation framework by \cite{Namugera2025contact} which models the epidemic dynamics in the three coupled networks. The initial infection density is set to \(5\%\) independently in each of the two graphs $\{G_A, G_B\}$. At each discrete time step, the infection probability at each node is updated according to the rules defined in Equations~\eqref{eq:Aupdate}--\eqref{eq:ABupdate}, which incorporate contributions from all nodes in a long-range network structure, distance-dependent transmission, stochastic environmental weights, and nonlinear co-infection dynamics.

Throughout, we impose $\delta > \mu$ with $\delta - \mu = 0.2$, ensuring that co-infection resolves more slowly than single infections. In \autoref{fig:fig1}, the prevalence of co-infection, $\rho_{AB}$, decreases monotonically with the co-infection recovery rate $\mu$, exhibiting a sharp transition at a critical value $\mu_c$ beyond which co-infection collapses and at least one strain goes extinct. This threshold behavior is characteristic of nonequilibrium phase transitions. The critical point $\mu_c$ is strongly influenced by network structure: smaller values of the long-range exponent $\alpha$ stabilize co-infection by sustaining transmission across distant nodes, whereas larger $\alpha$ amplify distance penalization and trigger abrupt extinction. Higher connectivity $\gamma$ further lowers $\mu_c$, reinforcing the destabilization. Together, these results show that long-range interactions promote persistence, while stronger recovery and spatial decay drive sharp epidemic thresholds.

\begin{figure}[ht]
    \centering
        \includegraphics[width=1\linewidth]{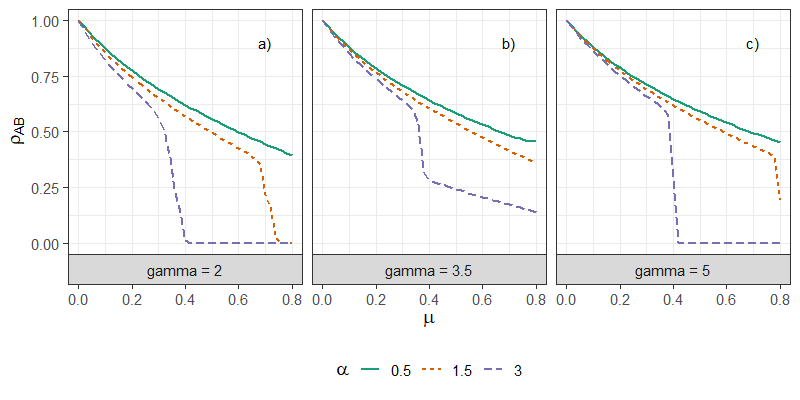}
        \caption{Phase transition for co-infection recovery rate. We consider a scale-free graph with $\sigma_K = 0.1$, $150$ nodes, $400$ vertices.}
    \label{fig:fig1}
\end{figure}

\begin{remark}[Remark 1:]
The co-infection prevalence decreases monotonically with the recovery rate~$\mu$. 
Formally, for $\epsilon>0$ and infection density $\rho$,
\[
\rho_{\mu}(\sigma, \gamma, \alpha) > 0 
\;\; \Rightarrow \;\;
\rho_{\mu - \epsilon}(\sigma, \gamma, \alpha) > 0, 
\quad \forall \epsilon > 0,
\]
which expresses that survival at a higher recovery rate implies survival at all lower ones.
\end{remark}

\begin{remark}[Remark 2:]
The co-infection density $(\rho_{AB})$, remains sensitive to the infection rates $\lambda_{1A}$ and $\lambda_{1B}$. Actually, the critical $\mu_c$ coincides with the critical $\delta_c = \inf \{\delta>0; \quad \rho_{AB} = 0 \}$ where there is epidemic extinction. As noted in ~\cite{Namugera2025contact,beutel2012interacting}, it is essential to define an interaction threshold that governs the persistence of the co-infection in regimes where $S_A > 1 > S_B$. In this case, the appropriate comparison is between the nearest-neighbor model and the present long-range kernel. We conjecture that the critical threshold for co-infection survival should be lower in the latter setting.

\end{remark}

\subsection{Competition interaction}
In this model, the simulation is simplified to two independent scale-free graphs, $G_A(n, p_1)$ and $G_B(n,p_2)$, each tracking the dynamics of one strain. We fix the recovery rate at $\delta = 0.5$ for both strains. We present the results as trajectory paths of both strains under three representative scenarios:  
\begin{enumerate}
    \item \textbf{Extinction}: $S_A, S_B < 1$, where neither strain persists.
    \item \textbf{Domination}: $S_A, S_B > 1$ but $\sigma_A \gg \sigma_B$, resulting in strain A outcompeting B.
    \item \textbf{Coexistence}: $S_A, S_B > 1$ and $\sigma_A \approx \sigma_B$, allowing both strains to persist.
\end{enumerate}

The trajectory arrows represent the dynamical flow of the two-strain system in phase space. Convergence to an axis corresponds to a stable fixed point where one strain dominates; convergence to the origin corresponds to the absorbing state of global extinction. Interior trajectories approaching a nonzero equilibrium indicate coexistence. The direction of the arrows reflects the vector field: vertical and horizontal motions correspond to single-strain growth or decay, while diagonal motions capture coupled dynamics of both strains.

We fix a scale-free graph of $250$ nodes and $800$ edges. The allocation and out-degree distribution is determined by the $\gamma$-value. In order to compute the score $(S_K)$, we require the largest eigenvalues of the transmission matrix in \autoref{tab:sfgraph_comp}.

\begin{table}[ht]
\centering
\begin{tabular}{c|ccc}
\hline
    $\gamma$ & \multicolumn{3}{c}{$\lambda_{1}$ for different $\alpha$} \\
    \cline{2-4}
 & $\alpha = 0.5$ & $\alpha = 1.5$ & $\alpha = 3.0$ \\
\hline
2.0 & 36.07 & 20.99 & 13.15 \\
3.5 & 25.70 & 13.24 & 7.32 \\
5.0 & 24.16 & 12.21 & 6.42 \\
\hline
\end{tabular}
    \caption{Largest eigenvalue (spectral radius) $\lambda_{1}$ estimate of the effective transmission matrix for different scale-free graph $G(250, 800)$ with exponents $\gamma$ and $\alpha$ (averaged over several Monte Carlo realizations).}
    \label{tab:sfgraph_comp}
\end{table}

Using these estimates, we determine the infection rates $\sigma_K$ for the simulation scenarios presented in \autoref{tab:scenarios}.

\begin{table}[ht!]
\centering
\begin{tabular}{lccc l}
\hline
Scenario    & $\sigma_A$ & $\sigma_B$ \\
\hline
Extinction  & 0.02 & 0.01  \\
Single strain dominance & 0.25 & 0.10 \\
Coexistence & 0.40 & 0.30  \\
\hline
\end{tabular}
\caption{Parameter settings for three illustrative epidemic scenarios.}
\label{tab:scenarios}
\end{table}

In \autoref{fig:extinction}, we examine the dynamics under weaker infection parameters. Panel (a) shows that for smaller values of~$\gamma$, both strains are more likely to go extinct. Extinction occurs for $\alpha \in \{1.5, 3\}$, while lower values of~$\alpha$ allow persistence due to the stronger role of long-range interactions. As $\gamma$ increases, this long-range effect becomes less pronounced, and the influence of~$\alpha$ on persistence is gradually reduced.

\begin{figure}[ht]
    \centering
    \includegraphics[width=1\linewidth]{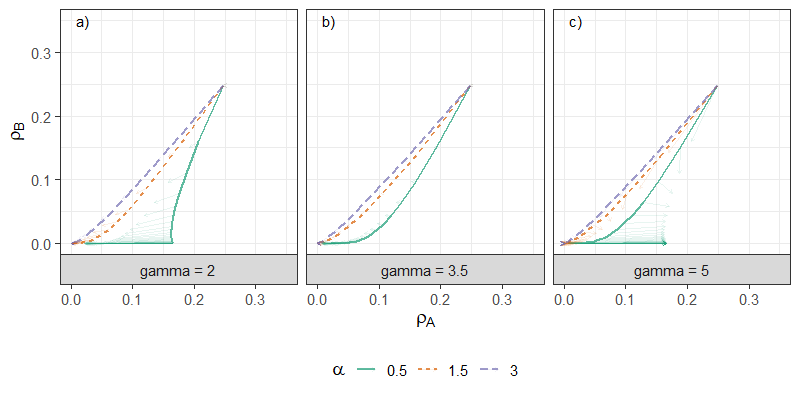}
    \caption{Extinction of the infection dominates across most of the parameter space. Simulations begin with the host population partially infected, evenly partitioned between the two strains, under a fixed value of $\delta = 0.5$, $\sigma_A = 0.03$ and $\sigma_B = 0.01$.}

    \label{fig:extinction}
\end{figure}

In scenario two shown in \autoref{fig:domination}, both strains are assigned $S>1$ to characterize their competitive strength in a heterogeneous environment. Low $\alpha$ values ($\alpha = 0.5, 1.5$) promote coexistence, whereas $\alpha > 2$ favours competitive exclusion. Arrows converging to the x-axis denote the complete extinction of strain B in favour of strain A, highlighting the asymmetry in strain dominance.

\begin{figure}[ht!]
    \centering
    \includegraphics[width=1\linewidth]{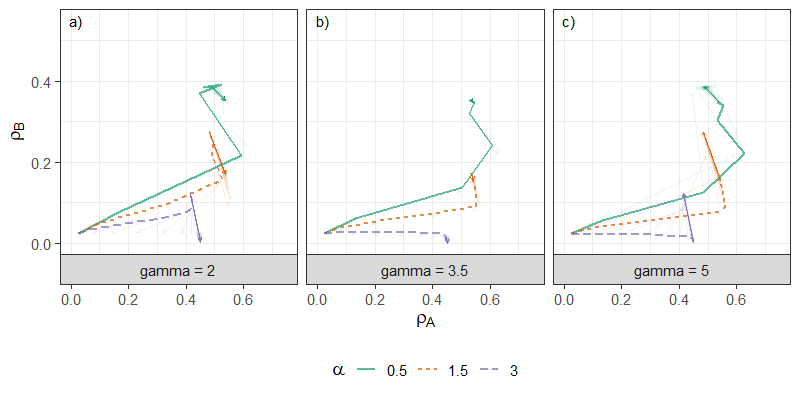}
    \caption{Single-strain domination. One strain outcompetes the other in extreme values of $\alpha$. Here, the trajectory converges to the axis corresponding to the dominant strain. Parameters: $\sigma_A = 0.25$, $\sigma_B = 0.1$.}
    \label{fig:domination}
\end{figure}

Finally, in the coexistence scenario (\autoref{fig:coexistance}), all strains persist at comparable densities. Environmental factors modulate the magnitude of strain prevalence but do not alter the overall coexistence pattern, indicating a robust and stable coexistence phase across the explored parameter space.

\begin{figure}[ht]
    \centering
    \includegraphics[width=1\linewidth]{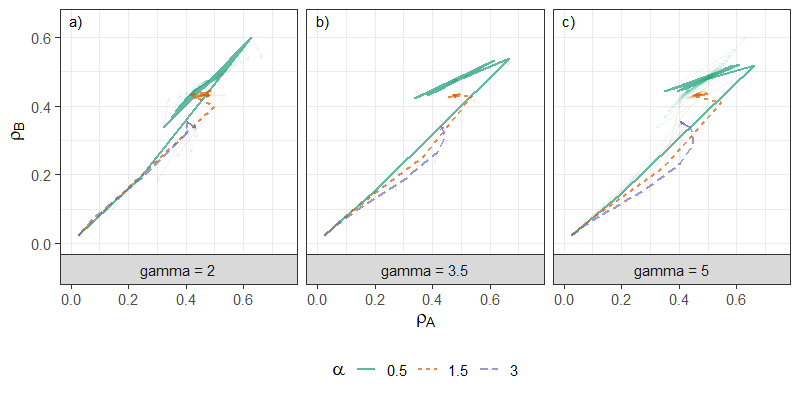}
    \caption{Both strains survive across all parameter regimes. Parameters: $\sigma_A = 0.4$, $\sigma_B = 0.3$ with $5\%$ of the nodes infected.}
    \label{fig:coexistance}
\end{figure}

\subsection{Statistical analysis}

To quantify the individual contributions of the parameters ($\alpha, \gamma,\mu$) to the observed dynamics, we fit statistical inference models to the simulated data and estimate their respective coefficients, providing a measure of each parameter's effect on strain competition outcomes.

\textbf{Co-infection:} Using data generated from numerical simulations, we fitted a Gaussian generalized linear model (GLM) to study the global effects of key parameters on the density of co-infection (\(\rho_{AB}\)) with \(\mu\), \(\alpha\) and \(\gamma\) as predictors (see \autoref{tab:glm_phiAB}). The intercept (\(1.065\)) represents the baseline level of \(\rho_{AB}\) when all predictors are zero.  The model shows that all three predictors are statistically significant.  
The parameter \(\mu\) has a strong negative effect (\(-0.979\), \(p<2\times10^{-16}\)), indicating that higher recovery of coinfected nodes reduces co-infection density. Similarly, \(\alpha\) has a significant negative effect (\(-0.109\), \(p<2\times10^{-16}\)), showing that steeper distance decay suppresses transmission. In contrast, \(\gamma\) exerts a small but significant positive effect (\(0.018\), \(p=0.00034\)), implying that greater degree heterogeneity favors coinfection.  Overall, the results highlight that \(\rho_{AB}\) is most sensitive to the recovery rate \(\mu\), followed by the spatial decay exponent \(\alpha\), with \(\gamma\) having a weaker but still significant contribution.

\begin{table}[ht]
\centering
\caption{GLM results for co-infection density ($\rho_{AB}$) with predictors $\mu$, $\alpha$, and $\gamma$. All coefficients are statistically significant. Other parameters: $\sigma = 0.1$.}
\label{tab:glm_phiAB}
\begin{tabular}{lcccc}
\hline
\textbf{Term} & \textbf{Estimate} & \textbf{Std. Error} & \textbf{t value} & \textbf{p-value} \\
\hline
Intercept & 1.0648 & 0.0236 & 45.198 & $<2\times10^{-16}$ *** \\
$\mu$     & -0.9787 & 0.0260 & -37.633 & $<2\times10^{-16}$ *** \\
$\alpha$  & -0.1094 & 0.0060 & -18.262 & $<2\times10^{-16}$ *** \\
$\gamma$  & 0.0182 & 0.0050 & 3.614 & 0.000344 *** \\
\hline
\multicolumn{5}{l}{\footnotesize Residual deviance = 5.1007, AIC = -522.67}
\end{tabular}
\end{table}

\paragraph{Competition(considering scenario 2):} 

The generalized additive model (GAM) used to study the nonlinear and time-dependent effects on the prevalence of the dominant strain \(A\) (\(\rho_A\)) reveals that its evolution is driven by strong, time-dependent effects of the network parameters \(\alpha\) and \(\gamma\). The intercept (\(0.495\)) as part of the results in \autoref{tab:GAM_phiA}, represents the baseline prevalence. 
The Effective Degrees of Freedom (EDF) signify the complexity of the smooth terms in the GAM, where an EDF close to $1$ indicates an almost linear effect. In contrast, higher EDF values indicate increasingly flexible and nonlinear relationships between the predictor and the response.
Both smooth terms are highly significant (\(p < 2\times10^{-16}\)) for 
\(s(t):\gamma\) (\(\mathrm{edf} = 9.8\)) which captures a complex, undulating dependence of prevalence on time modulated by the degree exponent \(\gamma\). Furthermore, \(s(t):\alpha\) (\(\mathrm{edf} = 5.85\)) describes a slightly simpler but equally strong influence of the spatial decay parameter \(\alpha\). The model explains approximately \(84.5\%\) of the variance (\(R^2_{\text{adj}} = 0.839\)), with a very small residual scale (\(\approx 0.0019\)), indicating an excellent fit.

These findings suggest that static effects of the network parameters cannot fully explain the spread of the dominant strain. Instead, \(\alpha\) and \(\gamma\) modulate the trajectory of infection at different stages of the epidemic. For instance, high-degree nodes (\(\gamma\)) dominate early transmission, while the spatial effect (\(\alpha\)) becomes increasingly important as the outbreak matures (refer to \autoref{fig:gam_fig} in the appendix). This model clarifies how the network topology and spatial structure sustain strain \(A\). The relation is similar for the competing strain \(B\). 

\begin{table}[ht]
\centering
\caption{Generalized additive model (GAM) results for the prevalence of strain \(A\) (\(\phi_A\)). 
Smooth terms capture the time-varying influence of the network parameters \(\alpha\) and \(\gamma\).}
\label{tab:GAM_phiA}
\begin{tabular}{lcccc}
\hline
Smooth term & edf & Ref.df & $F$ & p-value \\ \hline
$s(t):\gamma$ & 9.804 & 9.960 & 52.91 & $<2\times10^{-16}$ \\
$s(t):\alpha$ & 5.854 & 6.915 & 78.28 & $<2\times10^{-16}$ \\ \hline
\multicolumn{5}{l}{\footnotesize Adjusted $R^2 = 0.839$, Deviance explained = 84.5\%, n = 250, edges = 800.}
\end{tabular}
\end{table}

\section{Conclusion}

This work examined how environmental heterogeneity and long-range interactions shape the dynamics of symbiotic and competing two-strain epidemic models on complex networks. In the symbiotic case, we identified a critical co-infection recovery rate, $\mu_c$, whose value is strongly modulated by the interaction parameters $\alpha$ (spatial decay) and $\gamma$ (degree heterogeneity). Lower values of $\mu$, $\alpha$, and $\gamma$ favor epidemic persistence, whereas higher values raise the threshold $\mu_c$, suppressing co-infection.  A similar qualitative influence is observed for the competition model; however, environmental variability can additionally drive the exclusion of strains that would otherwise persist in isolation.

Statistical analyses further demonstrate that both $\alpha$ and $\gamma$ are significant determinants of epidemic growth, with $\alpha$ exerting the stronger effect. These findings underscore the importance of incorporating heterogeneous environments and long-range couplings into epidemic modeling and motivate a more rigorous characterization of phase transitions in multi-strain spreading processes.

\section*{Data Availability}
All data generated and analyzed during this study are available upon request from the corresponding author.

\section*{Conflicts of Interest}
The authors declare that there is no conflict of interest regarding the publication of this article.

\section*{Funding Statement}

This work has been supported by the Mathematics for Sustainable Development (MATH4SDG) project, which is a research and development project running in the period 2021-2026 at Makerere University-Uganda, University of Dar es Salaam-Tanzania, and the University of Bergen-Norway.

\section*{Acknowledgments}
I thank Ronald Katende, Nathan Muyinda, and the anonymous reviewers for their valuable comments during the writing of this paper.

{

}

\appendix
 \section{Appendix} 
We state the following Perron-Frobenius theorem.

\begin{theorem}[Perron-Frobenius]
Let $A \in \mathbb{R}^{n \times n}$ be a nonnegative and irreducible matrix. Then:  
\begin{enumerate}
    \item $A$ has a simple eigenvalue $\lambda_1 > 0$, such that $\lambda_1 = \max\{|\lambda| : \lambda \in \rho(A)\}$, where $\rho(A)$ denotes the spectrum of $A$.
    \item Every other eigenvalue $\lambda_i$ of $A$ satisfies $|\lambda| \leq \lambda_1$.
\end{enumerate}
\label{thm:pefro}
\end{theorem}

 \begin{figure}[ht]
     \centering
     \includegraphics[width=1\linewidth]{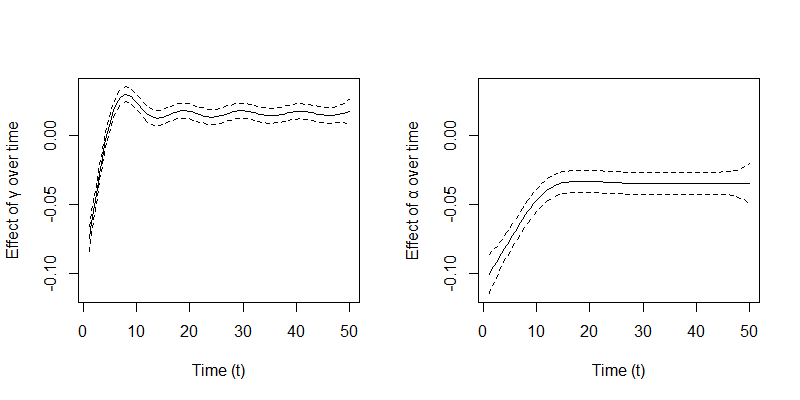}
     \caption{Generalized Additive Model (GAM) smooths for the effects of the parameters $\alpha$ and $\gamma$ on the prevalence of strain $A$ over time. The influence of $\gamma$ emerges almost immediately and stabilizes rapidly, whereas the effect of $\alpha$ develops more gradually before reaching a steady state.}
     \label{fig:gam_fig}
 \end{figure}
\end{document}